\documentclass[sigconf]{acmart}

\usepackage{booktabs} 

\usepackage{graphicx}
\usepackage{balance}  

\usepackage{amsmath}
\usepackage{amssymb}

\newtheorem{theorem}{Theorem}

\usepackage{algorithm2e}
\usepackage{color}
\usepackage{url}

\usepackage{cleveref}
\crefname{figure}{Figure}{Figures}
\Crefname{figure}{Figure}{Figures}
\crefname{section}{Section}{Sections}
\Crefname{section}{Section}{Sections}

\crefname{table}{Table}{Table}
\Crefname{table}{Table}{Table}

\newcommand{\simparam}[1]{\texttt{#1}}
\newcommand{\systemname}{\textsc{Credulix}}

\setcopyright{rightsretained}

\acmDOI{}

\acmISBN{}

\acmConference[]{}{}{}
\acmYear{}
\copyrightyear{}

\acmArticle{4}
\acmPrice{15.00}

\begin{document}
\title{Limiting the Spread of Fake News on Social Media Platforms by Evaluating Users' Trustworthiness}

\author{Oana Balmau}
\affiliation{%
  \institution{EPFL}
}
\email{oana.balmau@epfl.ch}

\author{Rachid Guerraoui}
\affiliation{%
  \institution{EPFL}
}
\email{rachid.guerraoui@epfl.ch}

\author{Anne-Marie Kermarrec}
\affiliation{%
  \institution{Mediego, EPFL}
}
\email{amk@mediego.com}

\author{Alexandre Maurer}
\affiliation{%
  \institution{EPFL}
}
\email{alexandre.maurer@epfl.ch}

\author{Matej Pavlovic}
\affiliation{%
  \institution{EPFL}
}
\email{matej.pavlovic@epfl.ch}

\author{Willy Zwaenepoel}
\affiliation{%
  \institution{EPFL}
}
\email{willy.zwaenepoel@epfl.ch}

\begin{abstract}

Today's social media platforms enable to spread both authentic and fake news very quickly.
Some approaches have been proposed to automatically detect such ``fake'' news based on their content,
but it is difficult to agree on universal criteria of authenticity (which can be bypassed by adversaries once known).
Besides, it is obviously impossible to have each news item checked by a human.

In this paper, we a mechanism to limit the spread of fake news which is not based on content.
It can be implemented as a plugin on a social media platform.
The principle is as follows: a team of fact-checkers reviews a small number of news items (the most popular ones),
which enables to have an estimation of each user's inclination to share fake news items.
Then, using a Bayesian approach, we estimate the trustworthiness of future news items,
and treat accordingly those of them that pass a certain ``untrustworthiness'' threshold.

We then evaluate the effectiveness and overhead of this technique on a large Twitter graph.
We show that having a few thousands users exposed to one given news item enables to reach a very precise estimation of its reliability.
We thus identify more than 99\% of fake news items with no false positives.
The performance impact is very small: the induced overhead on the 90th percentile latency is less than 3\%, and less than 8\% on the throughput of user operations.

\end{abstract}


%
%



\maketitle

\section{Introduction}

The expression ``fake news'' has become very popular after the 2016 presidential election in the United States. Both political sides accused each other of spreading false information on social media, in order to influence public opinion.   
Fake news have also been involved in Brexit and seem to have played a crucial role in the French election. The phenomenon is considered by many as a threat to democracy, since the proportion of people getting their news from social media is significantly increasing \cite{fakedemo}. 


Facebook and Google took a first concrete measure by removing advertising money from websites sharing a significant number of fake news \cite{fakead}. This, however, does not apply to websites that do not rely on such money: popular blogs, non-professional streaming channels, or media relying on donations, to name a few.  
Facebook also considered labeling some news items as ``disputed'' when independent human fact-checkers contest their reliability \cite{fakecheck}. However, there cannot be enough certified human fact-checkers for a worldwide social network. While it is very easy to share fake news, it may take very long to check them, clearly too long to prevent them from getting viral. 

We present \systemname, the first content-agnostic system to prevent fake news from getting \emph{viral}. 
From a software perspective, \systemname\ is a plugin to a social media platform.
From a more abstract perspective, it can also be viewed as a \emph{vaccine} for the social network. Assuming the system has been exposed to some (small) amount of fake news in the \emph{past}, \systemname\ enables it to prevent \emph{future} fake news from becoming viral.
It is important to note that our approach does not exclude other (e.g. content-based) approaches, but \emph{complements} them.

At the heart of our approach lies a simple but powerful Bayesian result we prove in this paper, estimating the credibility of news items based on which users shared them and how these users treated fake news in the past. News items considered fake with a sufficiently high probability can then be prevented from further dissemination, i.e., from becoming viral. 

Our Bayesian result is in the spirit of \emph{Condorcet's jury Theorem} \cite{austen1996information}, which states that a very high level of reliability can be achieved by a large number of weakly reliable individuals.
To determine the probability of falsehood of a news item X, we look at the behavior of users towards X. This particular behavior had a certain a priori probability to happen. We compute this probability, based on what we call 
\emph{user credulity records}: records of which fact-checked items users have seen and shared. Then, after determining the average fraction of fake news on the social network, we apply Laplace's Rule of Succession \cite{Zabell1989} and then Bayes' Theorem \cite{Koch1990} to obtain the desired probability. 


\systemname\ does retain the idea of using a team of certified human fact-checkers.
However, we acknowledge that they cannot review \emph{all} news items. Such an overwhelming task would possibly require even more fact-checkers than users. Here, the fact-checkers only check a \emph{few} \emph{viral} news items, i.e., ideally news items that have been shared and seen the most on the social network\footnote{Some news items are indeed seen by millions, and are easy to check a posteriori. For instance, according to CNN \cite{fakecnn}, the following fake news items were read by millions: \emph{``Thousands of fraudulent ballots for Clinton uncovered''}; \emph{``Elizabeth Warren endorsed Bernie Sanders''}; \emph{``The NBA cancels 2017 All-Star Game in North Carolina''}. According to  BuzzFeed \cite{fakebuzz}, the fake news item \emph{``Obama Signs Executive Order Banning The Pledge Of Allegiance In Schools Nationwide''}
got more than 2 millions engagements on Facebook, and
\emph{``Pope Francis Shocks World, Endorses Donald Trump for President, Releases Statement''}
almost 1 million.}.
Many such fact-checking initiatives already exist all around the world \cite{politifact,snopes,wpfactchecker,truthorfiction,fullfact}.

Such checks enable us to build \emph{user credulity records}.
Our Bayesian formula determines the probability that a given news item is fake using the records of users who viewed or shared it.
When this probability goes beyond a threshold (say 99.9999\%), the social network can react accordingly.
E.g., it may stop showing the news item in other users' news feeds.
It is important to note that our approach does not require any users to share a large amount of fake news.
It suffices that some users share more fake news than others.

Our approach is \emph{generic} in the sense that it does not depend on any specific criteria. Here, for instance, we look at what users \emph{share} to determine if a news item is \emph{fake}. However, the approach is independent of the precise meanings of ``share'' and ``fake'': they could respectively be replaced by (``like'' or ``report'') and  (``funny'', ``offensive'', ``politically liberal" or ``politically conservative'').

Turning the theory behind \systemname\ into a system deployable in practice is a non-trivial task.
In this paper we address these challenges as well.
In particular, we present a practical approach to computing news item credibility in a \emph{fast}, incremental manner.

We implement \systemname\ as a standalone Java plugin and connect it to Twissandra \cite{twissandra} (an open source Twitter clone),
which serves as a baseline system.
\systemname\ interferes very little with the critical path of users' operations and thus has a minimal impact on user request latency.
We evaluate \systemname\ in terms of its capacity to detect fake news as well as its performance overhead when applied to a real social network of over 41M users \cite{kwak2010twitter}.
After fact-checking the 1024 most popular news items (out of a total of over 35M items),
over 99\% of unchecked fake news items are correctly detected by \systemname.
We also show that \systemname\ does not incur significant overhead in terms of throughput and latency of user operations (sharing items and viewing the news feed): average latency increases by at most 5\%, while average throughput decreases by at most 8\%. 

The paper is organized as follows. \Cref{sec:theory} presents the theoretical principles behind \systemname. \Cref{sec:design} presents the design and implementation of \systemname. \Cref{sec:evaluation} reports on our evaluation results. \Cref{sec:discussion} discusses the limitations and tradeoffs posed by \systemname. \Cref{sec_rw} discusses related work and \cref{sec_conc} concludes.

\section{Theoretical Foundations}
\label{sec:theory}
In this section we give an intuition of the theoretical result underlying \systemname, followed by its formalization as a theorem.
We finally show how to restate the problem in a way that allows efficient, fast computation of news item credibility.

\begin{figure}
\begin{center}
\includegraphics[width=\columnwidth]{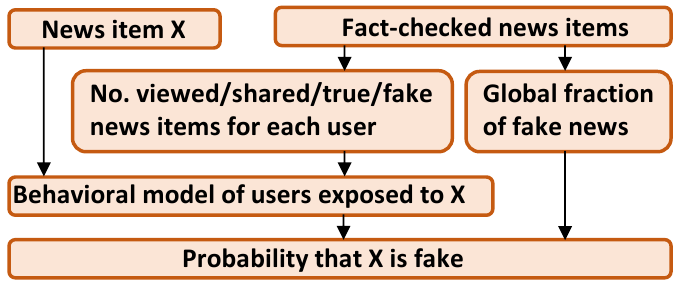}
\vspace{-.7cm}
\caption{News item falsehood probability computation.} 
\label{fig:schemebayes}
\end{center}
\vspace{-.5cm}
\end{figure}

\subsection{Intuition}
The context is a social network where users can post news items, such as links to newspapers or blog articles.
Users exposed to these news items can in turn share them with other users (Twitter followers, Facebook friends etc.). The social network has a fact-checking team whose role is to determine whether certain news items are fake (according to some definition of fake)\footnote{Our truth and falsehood criteria here are as good as the fact-checking team.}. The news items that the fact-checking team needs to check is very low compared to the total number of items in the social network.

\begin{figure}
\begin{center}
\includegraphics[width=\columnwidth]{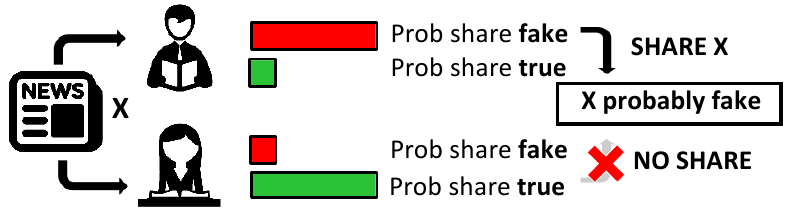}
\vspace{-.7cm}
\caption{Users reacting to new item $X$.} 
\vspace{-.5cm}
\label{fig:toyex}
\end{center}
\end{figure}

\vspace{.2cm}
\noindent\textbf{The Main Steps.}
\label{sec:main-steps} 
Our approach goes through the following three main steps: 
\begin{enumerate}
\item	The fact-checking team reviews few news items (ideally those that have been the most viral ones in the past). This is considered the ground truth in our context.
\item   \systemname\ creates a probabilistic model of each user's sharing behavior based on their reactions (share / not share) to the fact-checked items in Step (1).
  This captures the likelihood of a user to share true (resp. fake) news items.
\item	For a new, unchecked news item $X$, we use the behavior models generated in Step (2) to determine the probability that $X$ is fake, based on who viewed and shared $X$.
\end{enumerate}






A high-level view of our technique is depicted in \cref{fig:schemebayes}.
We use a Bayesian approach. For example, if an item is mostly shared by users with high estimated probabilities of sharing fake items, while users with high estimated probabilities of sharing true items rarely share it, we consider the item likely to be fake (\cref{fig:toyex}).

\vspace{.2cm}
\noindent\textbf{Preventing the Spread of Fake News.}
Once we estimate the probability of a news item $X$ being fake, preventing its spread becomes easy. Let $p_0$ be any cutoff probability threshold. Each time a user views or shares $X$, we compute $p$, the probability of $X$ being fake, which we detail below. If $p \geq p_0$ (i.e., $X$ has a probability at least $p_0$ to be fake),
\systemname\ stops showing $X$ in the news feed of users, preventing $X$ from spreading and becoming viral.

\subsection{Basic Fake News Detection}
\label{sec:detection-theory}

\vspace{.2cm}
\noindent\textbf{User Behavior.}
We model the behavior of a user $u$ using the two following probabilities:
\begin{itemize}
\item $P_T(u)$: probability that $u$ shares a news item if it is true.
\item $P_F(u)$: probability that $u$ shares a news item if it is fake.
\end{itemize}

The probabilities $P_T(u)$ and $P_F(u)$ are assumed to be independent between users. In practice, this is the case if the decision to share a news item $X$ is mainly determined by $X$ itself. 

We obtain estimates of $P_T(u)$ and $P_F(u)$ for each user based on the user's behavior (share / not share) with respect to fact-checked items. For any given user $u$, let $v_T(u)$ (resp. $s_T(u)$) denote the number of fact-checked true news items \emph{viewed} (resp. \emph{shared}) by $u$, and $v_F(u)$ (resp. $s_F(u)$) the number of fact-checked fake news items \emph{viewed} (resp. \emph{shared}) by $u$.
We call the tuple $(v_T(u), s_T(u), v_F(u), s_F(u))$ the \emph{User Credulity Record} (UCR) of $u$.

User behavior has been modeled similarly in prior work, for instance, in Curb \cite{Curb}. In Curb, users exposed to a news item decide probabilistically whether to share it, potentially exposing all their followers to it. Curb relies on very similar metrics, namely the numbers of viewed and shared items for each user, and the probabilities that users would share true or false news items.

For any given user $u$, we define the following functions, based on the UCR:

\begin{itemize}

\item $\beta_1(u) = ( s_T(u)+1 ) / ( v_T(u)+2 )$
\item $\beta_2(u) = ( s_F(u)+1 ) / ( v_F(u)+2 )$
\item $\beta_3(u) = ( v_T(u) - s_T(u) +1 ) / ( v_T(u)+2 )$
\item $\beta_4(u) = ( v_F(u) - s_F(u)+1 ) / ( v_F(u)+2 )$

\end{itemize}

According to Laplace's Rule of Succession \cite{Zabell1989}, we have $P_T(u) = \beta_1(u) $ and 
$P_F(u) = \beta_2(u) $.

\vspace{.2cm}
\noindent\textbf{Probability of a News Item Being Fake.}
\label{sec:fake-detection}
At the heart of \systemname\ lies a formula to compute the likelihood of a new (not fact-checked) news item $X$ to be fake. Let $V$ and $S$ be any two sets of users that have \emph{viewed} and \emph{shared} $X$, respectively. The probability that $X$ is fake, $p(V,S)$ is computed as:

\begin{equation}
\label{eq:formula}
\boxed{p(V,S) = g\pi_F(V,S)/(g\pi_F(V,S) + (1-g)\pi_T(V,S))}
\end{equation}

\noindent{Where:}
\begin{itemize}
\item $\pi_T(V,S) = \prod_{u \in S} \beta_1(u) \prod_{u \in V - S}  \beta_3(u)$
\item $\pi_F(V,S) = \prod_{u \in S} \beta_2(u) \prod_{u \in V - S} \beta_4(u)$, and
\item $g$ is the estimated \emph{global fraction} of fake news items in the social network, with $g \in$ $(0,1)$.
\end{itemize}

$g$ can be estimated by fact-checking a set of news items picked uniformly at random from the whole social network.
Let $g^*$ be the real fraction of fake news items in the social network. We distinguish two cases: $g^*$ is known, and $g^*$ is unknown. If $g^*$ is known, we have the following theorem.

\begin{theorem}
\label{thmprob}
Let $g = g^*$.
A news item \emph{viewed} by a set of users $V$ and \emph{shared} by a set of users $S$ is fake with probability $p(V,S)$.
\end{theorem}

\begin{proof} 

Consider a news item $X$ that has not been fact-checked. Consider the following events:
\begin{itemize}
\item[$E$]: $X$ viewed by a set of users $V$ and shared by a set of users $S$.
\item[$F$]: $X$ is fake.
\item[$T$]: $X$ is true.
\end{itemize}

Our goal is to evaluate $P(F|E)$: the probability that $X$ is fake knowing $E$.

\begin{itemize}

\item If $X$ is true, the probability of $E$ is $P(E|T) =$ \\$\prod_{u \in S} P_T(u) \prod_{u \in V - S} (1 - P_T(u)) = \pi_T(V,S)$.

\item If $X$ is fake, the probability of $E$ is
$P(E|F) =$
\\$\prod_{u \in S} P_F(u) \prod_{u \in V - S} (1 - P_F(u)) = \pi_F(V,S)$.

\end{itemize}

The probability that $X$ is fake (independently of $E$) is $P(F) = g^*$,
and the probability that $X$ is true (independently of $E$) is $P(T) = 1 - g^*$.

Thus, we can determine the probability that $E$ is true:
$P(E) = P(E|T)P(T) + P(E|F)P(F) = (1-g^*)\pi_T(V,S) + g^*\pi_F(V,S)$.

$P(F|E) = P(E|F)P(F)/P(E)$ according to Bayes' Theorem \cite{Koch1990}.
Then, $P(F|E) = g^*\pi_F(V,S)/(g^*\pi_F(V,S) + (1-g^*)\pi_T(V,S)) = p(V,S)$. Thus, the result.\end{proof}

If $g^*$ is unknown, we assume that $g$ is a lower bound of $g^*$. We get in this case the following theorem.

\begin{theorem}
\label{thmprob2}
For $g \leq g^*$, a news item \emph{viewed} by a set of users $V$ and \emph{shared} by a set of users $S$ is fake with probability \emph{at least} $p(V,S)$.
\end{theorem}

\begin{proof}

First, note that $\pi_T(V,S)$ and $\pi_F(V,S)$ are strictly positive by definition.
Thus, the ratio $\pi_T(V,S)/\pi_F(V,S)$ is always strictly positive.

$\forall x \in (0,1)$,
let $g(x) = x \pi_F(V,S) / ( x \pi_F(V,S) $ + $(1-x) \pi_T(V,S) )$.

Then, $p(V,S) = h(g)$, and 
according to Theorem~\ref{thmprob}, the news item is fake with probability $h(g^*)$. 

Written differently, $g(x) = 1 / (1 + k(x))$, with $k(x) = (1/x - 1) \pi_T(V,S) / \pi_F(V,S)$.

As $g \leq g^*$, $1/g \geq 1/g^*$, $1 + k(g) \geq 1 + k(g^*)$
and $h(g) \leq h(g^*)$. Thus, the result.\end{proof}

\subsection{Fast Fake News Detection}
\label{sec:computing-fake-probability}

\systemname' measure of credibility of a news item $X$ is the probability $p(V,S)$ that $X$ is fake.
An obvious way to compute this probability is to recalculate $p(V,S)$ using \cref{eq:formula} each time $X$ is viewed or shared by a user.
Doing so, however, would be very expensive in terms of computation.
Below, we show an efficient method for computing news item credibility.
We first describe the computation of UCRs,
and then present our fast, incremental approach for computing news item credibility using \emph{item ratings} and \emph{UCR scores}.
This is crucial for efficiently running \systemname\ in practice.

\vspace{.2cm}
\noindent\textbf{Computing User Credulity Records (UCRs).}
Recall that the four values $(v_T(u)$, $s_T(u),$ $v_F(u),$ $s_F(u))$ constituting a UCR only concern \emph{fact-checked} news items.
We thus update the UCR of user $u$ (increment one of these four values) in the following two scenarios.
\begin{enumerate}
\item When $u$ views or shares a news item that has been fact-checked (i.e., is known to be true or fake).
\item Upon fact-checking a news item that $u$ had been exposed to.
\end{enumerate}

In general, the more fact-checked news items a user $u$ has seen and shared, the more \emph{meaningful} $u$'s UCR. Users who have not been exposed to any fact-checked items cannot contribute to \systemname. 

\vspace{.2cm}
\noindent\textbf{Item Rating.}
In addition to $p(V,S)$, we introduce another measure of how confident \systemname\ is about $X$ being fake: the \emph{item rating} $\alpha(V,S)$, whose role is equivalent to that of $p(V,S)$.
We define it as $\alpha(V,S) = \pi_T(V,S)/\pi_F(V,S)$, $V$ and $S$ being the sets of users that viewed and shared $X$, respectively.
If we also define $\alpha_0 = (1/p_0 - 1) / (1/g - 1)$ as the rating threshold corresponding to the probability threshold $p_0$,
then, $p(V,S) \geq p_0$ is equivalent to $\alpha(V,S) \leq \alpha_0$.

We have $p(V,S) = g \pi_{F}(V,S) / ( g \pi_{F}(V,S) + (1-g) \pi_{T}(V,S) )
= g / ( g  + (1-g) (\pi_{T}(V,S) / \pi_{F}(V,S) ) ) =  g / ( g  + (1-g) \alpha(V,S) )$.
We have $p(V,S) \geq p_{0}$ if and only if $g/p(V,S) \leq g/p_0$, that is:
$g + (1-g) \alpha (V,S) \leq g/p_{0}$,
which is equivalent to
$\alpha (V,S) \leq (1/p_0 - 1)/(1/g - 1)$,
that is:
$\alpha (V,S) \leq \alpha_0$.

When the item $X$ with $\alpha(V,S) \leq \alpha_0$ is about to be displayed in a user's news feed, \systemname\ suppresses $X$.
Note that $\alpha_0$ can be a fixed constant used throughout the system, but may also be part of the account settings of each user,
giving users the ability to control how ``confident'' the system needs to be about the falsehood of an item before suppressing it.

According to the definition of $\pi_T(V,S)$ and $\pi_F(V,S)$,
each time $X$ is viewed (resp. shared) by a new user $u$,
we can update $X$'s rating $\alpha(V,S)$ by multiplying it by $\gamma_v(u) = \beta_1(u)/\beta_2(u)$ (resp. $\gamma_s(u) = \beta_3(u)/\beta_4(u)$).
We call $\gamma_v(u)$ and $\gamma_s(u)$ respectively the \emph{view score} and \emph{share score} of $u$'s UCR, as their value only depends on $u$'s UCR.
Consequently, when a user views or shares $X$, we only need to access a single UCR in order to update the rating of $X$. This is what allows \systemname\ to update news item credibility fast, without recomputing eq. (1) each time the item is seen by a user.

In what follows, we refer to $\gamma_v(u)$ and $\gamma_s(u)$ as $u$'s \emph{UCR score}.
The more a UCR score differs from $1$, the stronger its influence on an item rating (which is computed as a product of UCR scores).
We consider a UCR score to be \emph{useful} if it is different from $1$.

\section{\systemname\ as a Social Media Plugin}
\label{sec:design}

\begin{figure}[t]
\centering
\includegraphics[width=0.8\columnwidth]{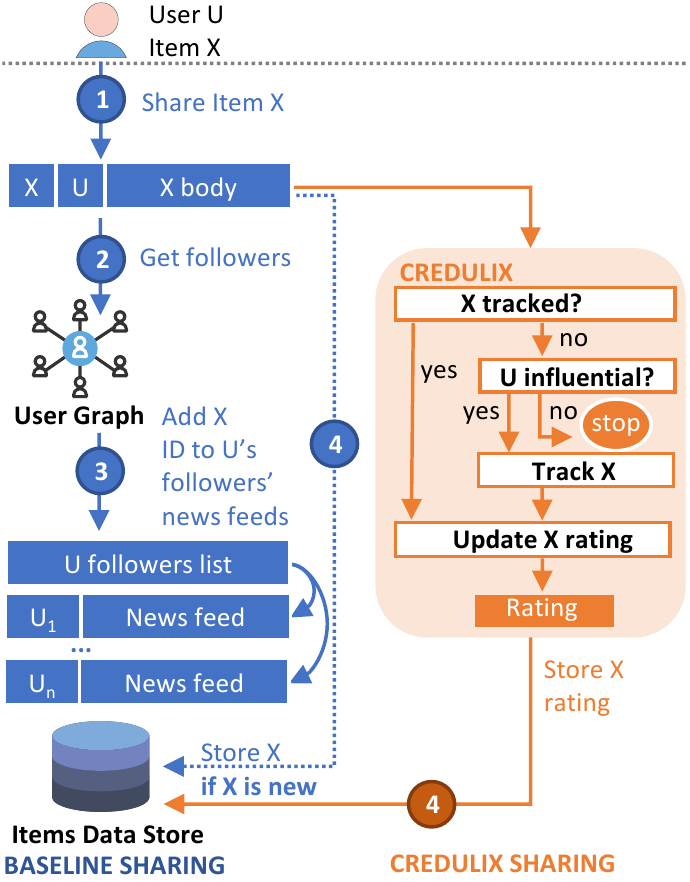}
\vspace{-.3cm}
\caption{CREDULIX' share operation.}
\vspace{-.5cm}
\label{fig:tweet-operation}
\end{figure}


\systemname\ can be seen as a plugin to an existing social network,
like, for instance, Facebook's translation feature.  The translator
observes the content displayed to users, translating it from one
language to another.  Similarly, \systemname\ observes news items
about to be displayed to users and tags or suppresses those considered
fake.

Despite the fast computtation described in \cref{sec:computing-fake-probability},
there are still notable challenges posed by turning an algorithm into a practical system.
In order for the \systemname\ plugin to be usable in practice, it must not impair user experience. 
In particular, its impact on the latency and throughput of user operations (retrieving news feeds or tweeting/sharing articles) must be small. Our design is motivated by minimizing \systemname' system resource overhead.






\vspace{.2cm}
\noindent\textbf{Selective Item Tracking.}
Every second, approximately 6000 new tweets appear on Twitter and 50000 new posts are created on Facebook \cite{internetstats,internetstats2}.
Monitoring the credibility of all these items would pose significant resource overhead. With \systemname, each view / share event requires an additional update to the news item's metadata.
However, we do not need to keep track of all the items in the system, but just the ones that show a potential of becoming viral.




\systemname\ requires each item's metadata to contain an additional bit indicating whether that item is \emph{tracked}.
The rating of item $X$ is only computed and kept up to date by \systemname\ if $X$ is tracked.

We set the \emph{tracked} bit for item $X$ when $X$ is shared by an \emph{influential user}.
We define influential users as users who have a high number of followers.
The intuition behind
this approach
is that a news item is more likely to become viral if it is disseminated by a well-connected user \cite{jenders13viraltweets}.
The follower threshold necessary for a user to be considered influential is a system parameter.
It can be chosen, for instance, as an absolute
threshold on the number of followers
or relatively to the other users (e.g., the first $n$ most popular users are influential). There are many methods to determine influential users \cite{lehmann2013finding,mathioudakis2009efficient,LearningInfluenceProbabilities, DiscoveringLeaders, RevenueMaximization, SocialInfluenceAnalysis, DiscoveringLeadersFromCA} or the likelihood of items to become viral \cite{ha2010news,ha2009relevance,lerman2010using,yang2015rain,lu2017predicting, ModelIndependentOL, TrackingInfluential, ScalableInfluenceMaximization}.
For simplicity, in our experiments we consider the 5\% most popular users influential.

\begin{figure}[t]
\centering
\includegraphics[width=0.8\columnwidth]{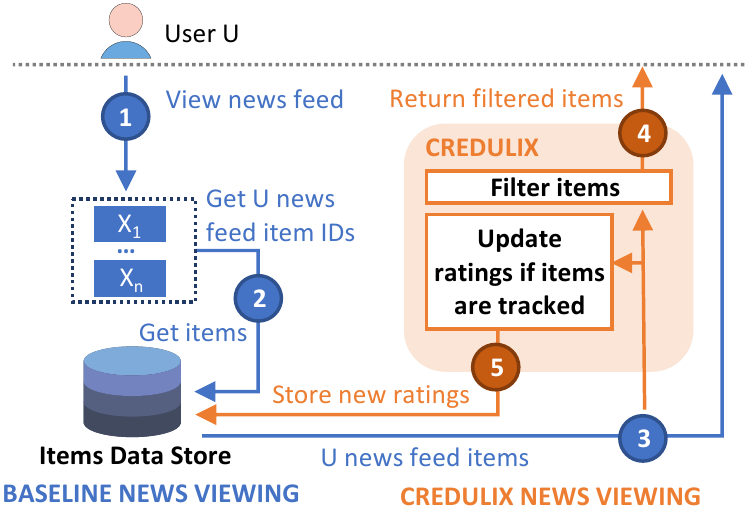}
\vspace{-.4cm}
\caption{CREDULIX' view operation.}
\label{fig:timeline-operation}
\vspace{-.5cm}
\end{figure} 


\vspace{.2cm}
\noindent\textbf{Interaction with the Social Media Platform.}
We consider two basic operations a user $u$ can perform:
\begin{itemize}
 \item \emph{Sharing} a news item and
 \item \emph{Viewing} her own news feed.
\end{itemize}
\emph{Sharing} is the operation of disseminating a news item to all of $u$'s followers (e.g., tweeting, sharing, updating Facebook status etc.).
\emph{Viewing} is the action of refreshing the news feed, to see new posts shared by users that $u$ follows.
In the following, we describe how these operations are performed in a social network platform (inspired by Twitter) without \systemname\ (Baseline) and with \systemname.
We assume that, like in Twitter, all users' news feeds are represented as lists of item IDs and stored in memory,
while the item contents are stored in an item data store \cite{twittertalk}.


\vspace{.2cm}
\noindent\textbf{Baseline Sharing.}
A schema of the \emph{Share} operation is shown in  \cref{fig:tweet-operation}.
The regular flow of the operation is shown in blue and \systemname\ is shown in orange.
User $u$ shares an item $X$ (1).
First, the social graph is queried to retrieve $u$'s followers (2).
The system then appends the ID of $X$ to the news feeds of $u$'s followers (3).
Finally, if $X$ is a new item,
the body of $X$ is stored in a data store (4).

\vspace{.2cm}
\noindent\textbf{Sharing with \systemname.}
If $u$ is not an influential user, the flow of the share operation described above stays the same.
If $u$ is influential, we mark $X$ as tracked and associate an item rating with $X$, because we expect $X$ to potentially become viral.
If $X$ is tracked, \systemname\ updates the rating of $X$ using $u$'s UCR share score. 
Thus, for tracked items, \systemname\ may require one additional write to the data store compared to the Baseline version,
in order to store the updated item rating.
This is done off the critical path of the user request, hence not affecting request latency.



\vspace{.2cm}
\noindent\textbf{Baseline News Feed Viewing.}
A schema of the \emph{View} operation is shown in \cref{fig:timeline-operation}.
User $u$ requests her news feed (1).
For each item ID in $u$'s news feed (stored in memory),
the system retrieves the corresponding item body from the data store (2)
and sends all of them back to the user (3).

\vspace{.2cm}
\noindent\textbf{Viewing News Feed with \systemname.}
\systemname\ augments the \emph{View} operation in two ways.
First, after the news feed articles are retrieved from the data store,
\systemname\ checks the ratings of the items,
filtering out the items with a high probability of being fake.
Second, if $u$'s news feed contains tracked items, \systemname\ updates the rating of those items using $u$'s UCR view score.
Hence, a supplementary write to the data store is necessary, compared to the Baseline version, for storing the items' updated ratings. Again, we do this in the background, not impacting user request latency.



\section{Evaluation}
\label{sec:evaluation}

In this section, we evaluate our implementation of \systemname\ as a stand-alone Java plugin.
We implement a Twitter clone where the \emph{share} and \emph{view} operation executions are depicted in \cref{fig:tweet-operation,fig:timeline-operation}.
We refer to the Twitter clone as \emph{Baseline}
and we compare it to the variant with \systemname\ plugged in, which we call \systemname.
For the data store of the Baseline we use Twissandra's data store \cite{twissandra},
running Cassandra version 2.2.9 \cite{cassandra}.

The goals of our evaluation are the following. First, we explore \systemname's fake news detecting efficiency. Second, we measure the performance overhead of our implementation. More precisely, we show that:
\begin{enumerate}
\item \systemname\ efficiently stops the majority of fake news from becoming viral, with no false positives. \systemname\ reduces the number of times a viral fake news item is viewed from hundreds of millions to hundreds of thousands (in \cref{sec:exp2}).
\item \systemname\ succeeds in stopping the majority of fake news from becoming viral for various user behaviors in terms of how likely users are to share news items they are exposed to (in \cref{sec:exp3}).
\item \systemname's impact on system performance is negligible for both throughput and latency (in \cref{sec:exp4}).
\end{enumerate}

\subsection{Experimental Setup}
\label{sec:setup}

We perform our evaluation using a real Twitter graph of over 41M users \cite{kwak2010twitter}. We consider users to be influential (as defined in \cref{sec:design}) if they are among the 5\% most followed users.
We use a set of tweets obtained by crawling Twitter to get a distribution of item popularity.
Out of over 35M tweets we crawled, the 1024 (0.003\%) most popular tweets are retweeted almost 90 million times, which corresponds to over 16\% of all the retweets.

Two key values influence \systemname' behavior:
\begin{itemize}
  \item \textbf{\simparam{r}:} The number of fact-checked news items during UCR creation (i.e., the number of news items constituting the ground truth).
    In our experiments we use a value of $\simparam{r} = 1024$, which causes one third of the user population to have useful UCR scores
    (more than enough for reliable fake item detection).
  \item \textbf{\simparam{msp}:} The max share probability models users' intrinsic sharing behavior: how likely users are to share news items they are exposed to. This models how users react to news items. It is not a system parameter of \systemname. We expect \simparam{msp} to be different for different news items, as some items become viral (high \simparam{msp}) and some do not (low \simparam{msp}).
Regardless of what the real value of \simparam{msp} is, \systemname\ effectively prevents fake items from going viral, as we show later.
\end{itemize}

While the network and the tweets come from real datasets, we generate the user behavior (i.e., probability to share fake and true news items), as we explain in the remainder of this section. 
We proceed in two steps. 
\begin{enumerate}
	\item \emph{UCR creation}: determining the UCR (i.e., $v_T, s_T, v_F, s_F$) for each user based on propagation of fact-checked news items.
	\item \emph{Fake item detection}: using the UCRs obtained in the previous step, we use \systemname\ to detect fake news items and stop them from spreading.
\end{enumerate}
This two-step separation is only conceptual.
In a production system, both UCR creation and fake item detection happen continuously and concurrently.

\vspace{.2cm}
\noindent\textbf{UCR Creation.}
For each user $u$, we set $P_T(u)$ and $P_F(u)$ (see \cref{sec:theory}) to values chosen uniformly at random between $0$ and \simparam{msp}. The likelihood of a user to share true or fake news is the main user characteristic used by \systemname. This approach yields different types of user behavior. We take a subset of \simparam{r} tweets from our tweet dataset and consider this subset the ground truth, randomly assigning truth values to news items.

To create the UCRs, we propagate \simparam{r} items through the social graph.
We assign each of the \simparam{r} items a target share count, which corresponds to its number of retweets in our dataset. The propagation proceeds by exposing a random user $u$ to the propagated item $X$ and having $u$ decide (based on $P_T(u)$ and $P_F(u)$) whether to share $X$ or not. If $u$ shares $X$, we show $X$ to all $u$'s followers that have not yet seen it. During item propagation, we keep track of how many true/fake items each user has seen/shared and update the UCRs accordingly. 

We repeat this process until one of the following conditions is fulfilled:
\begin{enumerate}
\item The target number of shares is reached.
\item At least 80\% of users have been exposed to $X$, at which point we consider the network saturated.
\end{enumerate}

\vspace{.2cm}
\noindent\textbf{Fake Item Detection.}
After creating the UCRs, we measure how effectively these can be leveraged to detect fake news.
To this end, in the second step of the evaluation, we propagate tweets through the social graph.
One such experiment consists of injecting an item in the system, by making a random user $u$ share it.
The propagation happens as in the previous phase, with two important differences:
\begin{enumerate}
\item Since in this second step the item we propagate is not fact-checked, we do not update $u$'s UCR.
Instead, whenever $u$ is exposed to an item, we update that item's rating using $u$'s UCR score.
We use the share score if $u$ shares the item, otherwise we use the view score (see \cref{sec:theory}).
\item We only propagate the item once, continuing until the propagation stops naturally, or until the probability of an item being fake reaches $p_0 = 0.999999$.
\end{enumerate}
We repeat this experiment $500$ times with a fake news item and $500$ times with a true news item to obtain the results presented later in this section.

We conduct the experiments on a 48-core machine,
with four 12-core Intel Xeon E7-4830 v3 processors operating at 2.1 GHz,
512 GB of RAM, running Ubuntu 16.04.

\subsection{Stopping Fake News from Becoming Viral}
\label{sec:exp2}

This experiment presents the end-to-end impact of\\
 \systemname\ on the number of times users are exposed to fake news. To this end, we measure the number of times a user is exposed to a fake news item in two scenarios:
\begin{enumerate}
\item Baseline propagation of fake news items.
\item Propagation of fake items, where we stop propagating an item when \systemname\ labels it as fake with high probability.
\end{enumerate}
\Cref{fig:msp-viewCnt} conveys results for items with varying rates of virality, modeled by our \simparam{msp} parameter. It shows how many times a user has been exposed to a fake item, cumulatively over the total number of fake items that we disseminate. We can see that regardless of how viral the items would naturally become, \systemname\ is able to timely detect the fake items before they spread to too many users. \systemname\ restricts the number of views from hundreds of millions to tens or hundreds of thousands.

None of our experiments encountered false positives (i.e., true items being incorrectly labeled as fake).
Considering the increasing responsibility being attributed to social network providers as mediators of information, it is crucial that true news items are not accidentally marked as fake.



\subsection{Fake News Detection Relative to Sharing Probability}
\label{sec:exp3}

\begin{figure}
\centering
\includegraphics[width=\columnwidth]{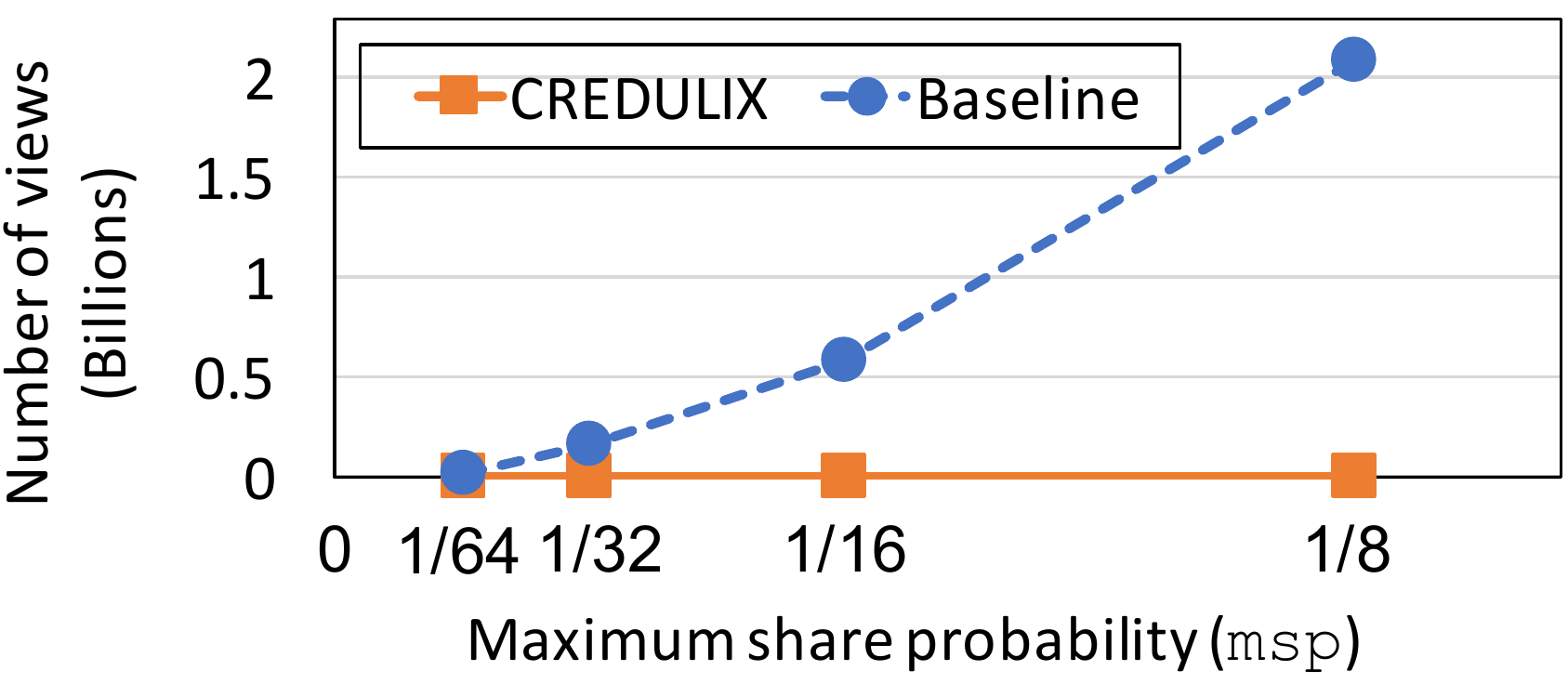}
\caption{Fake news spreading with CREDULIX, as a function of \simparam{msp} (lower is better). For low \simparam{msp}, news items do not become viral. For high \simparam{msp}, CREDULIX blocks the majority of fake news items.}
\label{fig:msp-viewCnt}
\end{figure}

In \cref{fig:msp-relative} we plot the percentage of fake items displayed with \systemname\ for two graph sizes.
On a smaller graph of 1M users generated with the SNAP generator \cite{snapnets},
\systemname\ achieves a lower fake item detection rate. This is because the impact of fact-checked items is smaller on a small graph, leading to fewer users with relevant UCR scores.
This result suggests that on a real social graph that is larger than the one we use, \systemname\ would be more efficient than in our experiments.

\Cref{fig:msp-relative} also shows how the detection rate depends on the tendency of users to share news items (that we model using \simparam{msp}).
The more viral the items get (the higher the \simparam{msp} value), the more effective \systemname\ becomes at fake item detection.
Intuitively, the more items users share, the more precisely we are able to estimate their sharing behavior.
The lower detection rate for small \simparam{msp} values does not pose a problem in practice, as a low \simparam{msp} also means that items naturally do not become viral.

While not visible in the plot, it is worth noting that not only the relative amount of viewed fake items decreases, but also the \emph{absolute} one.
For example, while for \simparam{msp} $= 1/32$ a fake news item has been displayed almost 3k times (out of over 84k for Baseline),
for \simparam{msp} $= 1/16$ a fake item has only been displayed 1.2k times (out of over 128k Baseline) in the 1M graph.
Interestingly, with increasing tendency of items to go viral (i.e. increasing \simparam{msp}),
even the \emph{absolute} number of displayed fake items decreases.
The relative decrease effect is not due to an absolute increase for Baseline.
Instead, it is due to a higher \simparam{msp} value ensuring more spreading of (both true and fake) news items in our UCR creation phase.
This in turn produces better UCRs, increasing \systemname' effectiveness.

\begin{figure}
\centering
\includegraphics[width=\columnwidth]{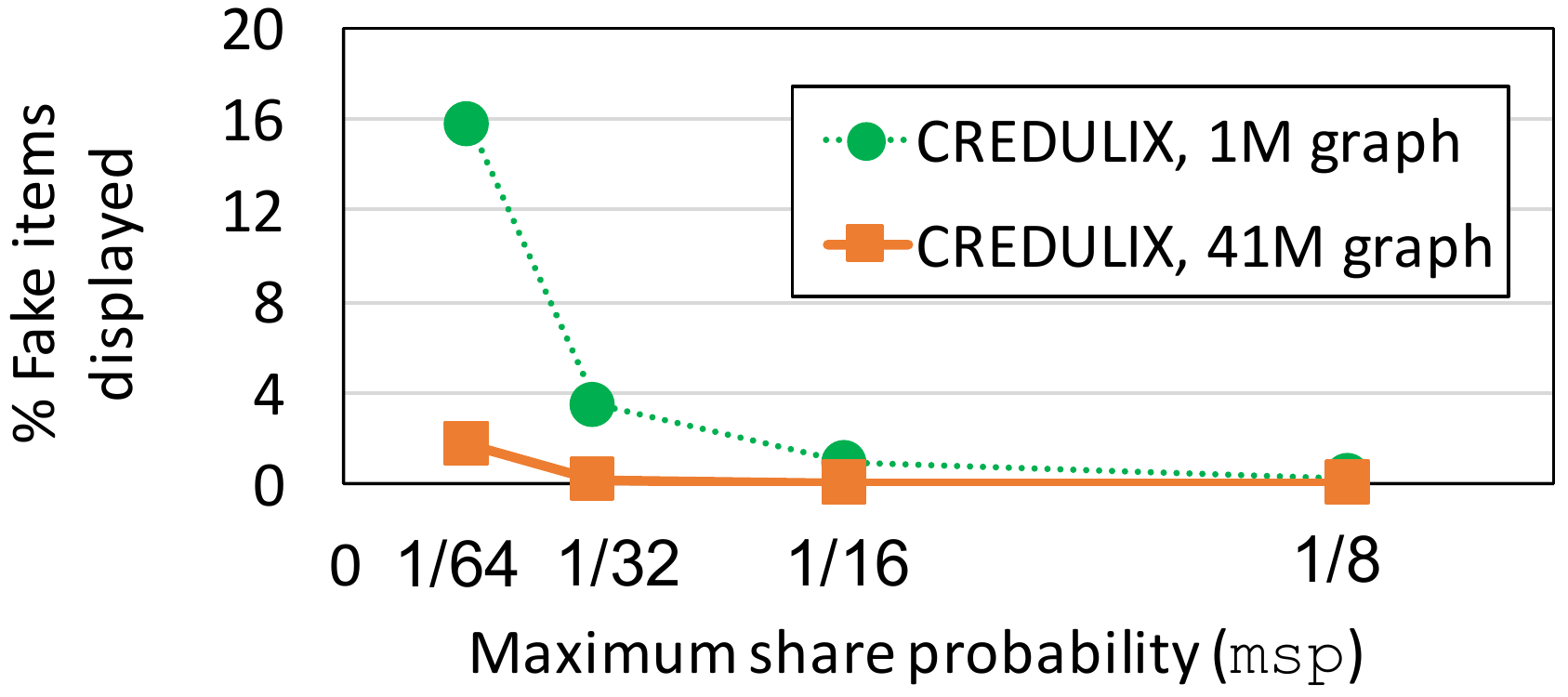}
\caption{Fake news spreading with CREDULIX, as a function of \simparam{msp}, for different social graph sizes (lower is better).}
\label{fig:msp-relative}
\end{figure}

\subsection{CREDULIX Overhead}
\label{sec:exp4}

\begin{table}[b]
\centering
\bgroup
\begin{tabular}{|c|c|}
\hline
\textbf{\simparam{msp} Value} & \textbf{\% Views, \% Shares} \\
\hline
1/8 & 94\% Views, 6\% Shares \\
\hline
1/16 & 97\% Views, 3\% Shares \\
\hline
1/32 & 99\% Views, 1\% Shares \\
\hline
1/64 & 99.9\% Views, 0.1\% Shares \\
\hline
\end{tabular}
\egroup
\vspace{.2cm}
\caption{Workload Characteristics. The only parameter we vary is \simparam{msp}, from which the view/share ratio follows.}
\vspace{-.5cm}
\label{tab:workloads}
\end{table}

\begin{figure}
\centering
\includegraphics{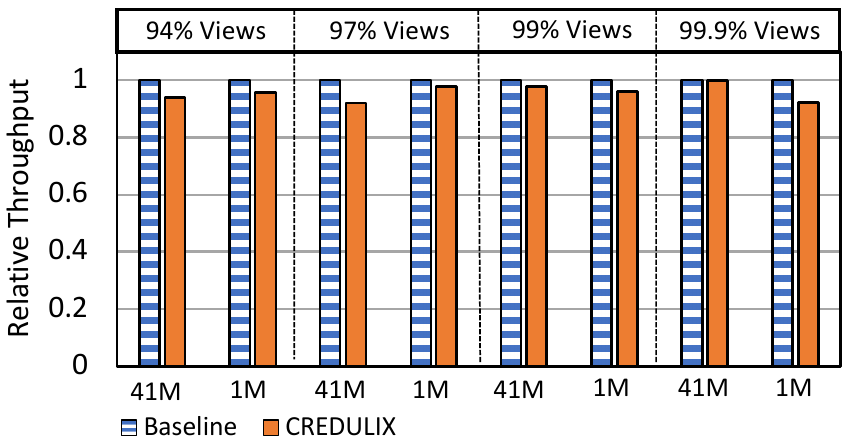}
\vspace{-.7cm}
\caption{\systemname's throughput overhead}
\label{fig:system-eval-throughput}
\end{figure} 

In this experiment, we evaluate \systemname' impact on user operations' (viewing and sharing) \emph{throughput} and \emph{latency}.
We present our results for four workloads, each corresponding to a value of \simparam{msp} discussed above (\cref{fig:msp-viewCnt,fig:msp-relative}).
The four workloads are summarized in \cref{tab:workloads}.
We present results for two social graph sizes: 41M users, and 1M users, with 16 worker threads serving user operations, showing that the \systemname's overhead in terms of throughput and latency is low.
 

\Cref{fig:system-eval-throughput} shows the throughput comparison between \systemname\ and Baseline, for the four workloads.
The throughput penalty caused by \systemname\ is at most 8\%. The impact on throughput is predominantly caused by \systemname' background tasks, as detailed in \cref{sec:design}.
Moreover, \systemname\ does not add significant overhead relative to the Baseline as the graph size increases.
The throughput differences between the two graph sizes are not larger than 10\%.
This is due to our design which relies on selective item tracking.

\Cref{fig:system-eval-latency-6shares,fig:system-eval-latency-1shares} show view and share latencies for the 94\% views and 99.9\% views workloads, respectively. The latency values for the two other workloads are similar and we omit them for brevity.

For the 41M User Twitter graph, the average and 90th percentile latencies are roughly the same for \systemname\ and for Baseline. We notice, however, heavier fluctuations for the 1M User graph.  
Overall, latency increases by at most 17\%, at the $90^{th}$ percentile, while the median latency is the same for both operations, for both systems, for both graphs (4 microseconds per operation).
The low overhead in latency is due to \systemname\ keeping its computation  outside the critical path.
Standard deviation of latencies is high both for the Baseline and for \systemname, for both share and view operations.
The high variation in latency is caused by the intrinsic differences between the users;
for instance, share operations of a user with more followers need to propagate to more users than posts of users with few or no followers.

The high $99^{th}$ percentile latency for both systems results from Twissandra (our Baseline) being implemented in Java,
a language (in)famous for its garbage collection breaks.
Certain operations, during which garbage collection is performed by the JVM, thus experience very high latencies.
This effect is not inherent to the system itself,
but is caused by the Java implementation and can be avoided by using a language with explicit memory management.
The impact of garbage collection is stronger with \systemname\ than with Baseline,
as \systemname\ creates more short-lived objects in memory to orchestrate its background tasks.
In addition to the intrinsic differences between users discussed above,
garbage collection also significantly contributes to the high standard deviation observed in all latencies.

\begin{figure}
\centering
\includegraphics{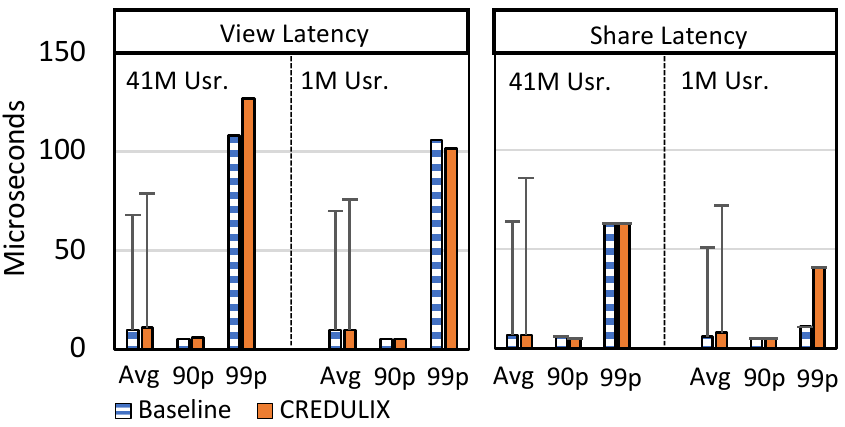}
\vspace{-.7cm}
\caption{\systemname's latency overhead: 94\% views}
\label{fig:system-eval-latency-6shares}
\end{figure} 

\begin{figure}
\centering
\includegraphics{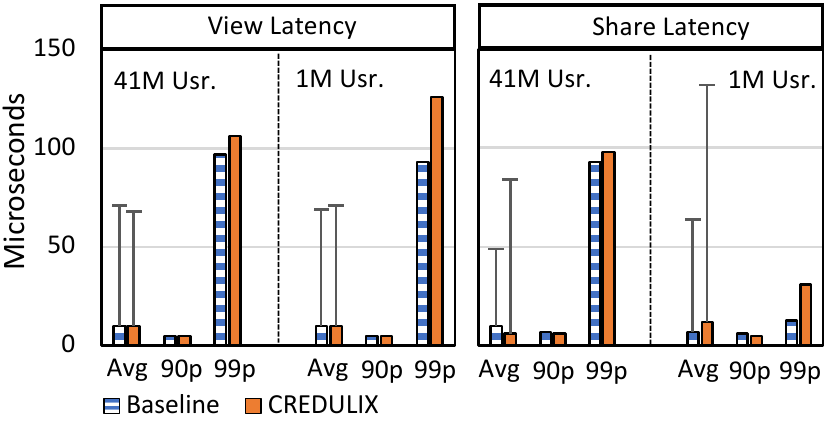}
\vspace{-.7cm}
\caption{\systemname's latency overhead: 99.9\% views}
\label{fig:system-eval-latency-1shares}
\end{figure}




\section{Discussion and limitations}
\label{sec:discussion}

We believe that \systemname\ is a good step towards addressing the fake news problem, but we do not claim it to be the ultimate solution. \systemname\ is one of many possible layers of protection against fake news and can be used independently of other mechanisms.
With \systemname\ in place, news items can still be analyzed based on content using other algorithms.
It is the combination of several approaches that can create a strong defense against the fake news phenomenon.
This section discusses the limitations of \systemname.

\vspace{2mm}
\noindent\textbf{News Propagation.}
\systemname\ does not prevent users from actively pulling any (including fake) news stories directly from their sources.
\systemname\ \emph{identifies} fake news on a social media platform and,
if used as we suggest,
prevents users from being notified about other users sharing fake news items.

\vspace{2mm}
\noindent\textbf{Manual Fact-Checking.}
\systemname\ relies on manual fact-checking and thus can only be as good as the fact-checkers.
Only users who have been exposed to manually fact-checked items can be leveraged by \systemname.
However, fact-checking a small number of popular news items is sufficient to obtain enough users with usable UCRs.
Fact-checking a small number of news items is feasible, especially given the recent upsurge of fact-checking initiatives \cite{politifact,snopes,wpfactchecker,truthorfiction,fullfact}.

\vspace{2mm}
\noindent\textbf{User Behavior.}
\systemname's algorithm is based on the assumption that among those users exposed to fact-checked news items, some share more fake items than others.
Analogous assumptions are commonly used in other contexts such as recommender systems.
For example, a user-based recommender system would not be useful if all users behaved the same way, i.e. everybody giving the same ratings to the same items.
Note that, even in the case where all users did behave the same, running \systemname\ would not negatively impact the behavior of the system: \systemname\ would not detect any fake news items, nor would it classify fake news items as true.
Our approach shines when the inclination of most users to share fake news does not change too quickly over time.
Looking at social media today, some people seem indeed to consciously and consistently spread more fake news than others \cite{macedonia,macron,gullible}.
In fact, many systems that are being successfully applied in practice (e.g. reputation systems, or systems based on collaborative filtering) fundamentally rely on this same assumption.

\vspace{2mm}
\noindent\textbf{Malicious Attacks.}
The assumption that users do not change their behavior too quickly could, however, potentially be exploited by a malicious adversary.
Such an adversary controlling many machine-operated user accounts could deceive the system by breaking this assumption.
For example, all accounts controlled by the adversary could be sharing only true reviewed news items for an extended period of time and then suddenly share an un-reviewed fake one (or do the opposite, if the goal is to prevent a truthful news item from being disseminated).
Even then, a successful attack would only enable the spread of the news item, without guaranteeing its virality.
Moreover, the adversary runs a risk that the fake item will later be fact-checked and thus will appear in their UCRs, reducing the chances of repeating this attack with another fake item.
In fact, the more popular such an item becomes (which is the likely goal of the adversary), the higher the chance of it being fact-checked.
The trade-off between false positives and false negatives is expressed by the $p_0$ parameter (see \cref{sec:theory}), i.e. the certainty required to flag an item as fake.
A high value of $p_0$ might make it easier for the adversary to ``smuggle'' fake items in the system, but makes it more difficult to prevent true news items from spreading.

\vspace{2mm}
\noindent\textbf{Updating of the Ground Truth.}
Like any vaccine, \systemname\ relies on a fraction of fake news to exist in the social network in order to be efficient. If \systemname\ stops the fake news from becoming viral, then the system might lack the ground truth to make future predictions. Hence, there might be periods when fake news can appear again.
To avoid such fluctuations, \systemname' ground truth should be continuously updated with some of the most current fake news items.
\systemname' evolution in time, including changes in user behavior as well as updating of the ground truth related to system dynamics, are research directions we are considering for future work.

\vspace{2mm}
\noindent\textbf{Filtering News.}
One could argue that removing some news items from users' news feeds might be seen as a limitation, even as a form of censorship.
But social media already take that liberty as they display to users only about 10\% of the news that they could show.
Rather than censorship, \systemname\ should be viewed as an effort to ensure the highest possible quality of the items displayed, considering the credibility of an item to be one of the quality criteria.


\section{Related Work}
\label{sec_rw}

\systemname\ shares similarities with \emph{reputation systems} \cite{recomsysw,Kokkodis:2013,rodriguez2007smartocracy},
in creating profiles for users (UCRs in \systemname) and in assuming that the future behavior of users will be similar to their past behavior.
In our approach, however, users are not rated directly by other users.
Instead, we compute users' UCRs based on their reaction to what we consider ground truth (fact-checked items).

\systemname\ also resembles \emph{recommender systems} \cite{resnick1997recommender,ahn2010towards,amatriain2009collaborative,lin2012premise,liu2011wisdom,su2009survey,garimella2017mary,koutra2015events,de2012chatter,onuma2009tangent} in the sense that it pre-selects items for users to see.
Unlike in recommender systems, however, the pre-selection is independent of the requesting user.
Our goal is not to provide a personalized selection.

Another approach to detect fake news is to \emph{automatically check content} \cite{fake_T6}. Content analysis can also help detect machine-generated fake blogs \cite{fake_S2} or social spam using many popular tags for irrelevant content \cite{fake_S3}.
Another line of research has been motivated by the role of social networks for news dissemmination in the case of catastrophic events such as hurricanes and earthquakes\cite{fake_U1,fake_N4}.

News item credibility can also be inferred by applying \emph{machine learning} techniques \cite{dewan2015towards},
by using a set of reliable content as a training set \cite{fake_M2},
or by analyzing a set of predetermined features \cite{fake_N2}.
Other parameters of interest are linguistic quantifiers \cite{fake_T}, swear words, pronouns or emoticons \cite{fake_M3}.
Yet, a malicious agent knowing these specific features could use them to spread fake news.

Following the events of 2016, Facebook received much media attention concerning their politics about ``fake news''.
Their first aproach was to \emph{assess} \emph{news sources reliability} in a centralized way \cite{fb_wpost}. Recently, Facebook launched \emph{community-based} assesment experiment \cite{fb_techcrunch}:
asking the users to evaluate the reliability of various news sources.
The idea is to give more exposure to news sources that are ``broadly trusted'''. Our approach is finer-grained and goes to the level of news \emph{items}. Facebook also used third-party fact checkers to look at articles flagged by users.

Dispute Finder \cite{DisputeFinder} also takes a community approach asking users not only to flag claims they believe to be disputed, but also to link the disputed claim to a reliable source that refutes it. \systemname, does not look at concrete claims, but infers the trustworthiness of news items based on user behavior.

\emph{Fact-checking tools} can help to annotate documents and to create knowledge bases \cite{DocumentCloud, OpenCalais, FactMinder, SourceSight, FusingData}. Automated verification of claims about structured data, or finding interesting claims that can be made over a given set of data has been researched as well \cite{iCheck, FactWatcher}.
These tools facilitate the fact-checking process that \systemname\ relies on, but by themselves do not prevent the spreading of fake news. Curb \cite{Curb}, is algorithm that is close to \systemname\ in several aspects. Like \systemname, it leverages the crowd to detect and reduce the spread of fake news and misinformation and  assumes a very similar user behavior model. Curb, however, focuses on the problem of which items to fact-check and when, relying on users to manually flag items. Curb only prevents the spreading of items that have been fact-checked.
In addition, it assumes fact-checking to happen instantaneously,
i.e., a fake item selected for fact-checking by Curb is assumed to stop spreading immediately,
without taking into account the considerable fact-checking delay. \systemname\ could benefit from Curb by using it as a fact-checking module.



\section{Conclusions}
\label{sec_conc}

We presented \systemname, the first content-agnostic system to detect and limit the spread of fake news on social networks with a very small performance overhead.
Using a Bayesian approach, \systemname\ \emph{learns} from human fact-checking to compute the probability of falsehood of news items,
based on who shared them.
Applied to a real-world social network of 41M users, \systemname\ prevents the vast majority (over 99\%) of fake news items from becoming viral, with no false positives.


\bibliographystyle{ACM-Reference-Format}
\bibliography{sample-bibliography}

\end{document}